\documentclass[12pt]{article}
\usepackage[utf8]{inputenc}
\usepackage{preamble}
\pdfoutput=1

\def\sym{\mathsf{sym}}

\def\Vand{\mathsf{Vand}}
\def\bias{\mathsf{bias}}

\newcommand{\summ}[1]{%
  \begingroup\noexpandarg
  \StrLen{#1}[\temp]%
  1^\top
  \ifnum\temp>1
    (#1)%
  \else
    #1% 
  \fi 
  \endgroup
}

\def\toggleComments{0}
\ifnum\toggleComments=1
\newcommand{\chin}[1]{\footnote{{\bf \color{cyan}Chin}: {#1}}}
\newcommand{\manu}[1]{\footnote{{\bf \color{purple}Manu}: {#1}}}
\newcommand{\peter}[1]{\footnote{{\bf \color{green}Peter}: {#1}}}
\fi
\ifnum\toggleComments=0
\newcommand{\chin}[1]{}
\newcommand{\manu}[1]{}
\newcommand{\peter}[1]{}
\fi

\def\epsilon{\varepsilon}

\title{Pseudorandomness, symmetry, smoothing: II}

\author{%
Harm Derksen\thanks{Partially supported by NSF grant DMS 2147769.} \\
Northeastern University\\
ha.derksen@northeastern.edu
\and
Peter Ivanov\thanks{Supported by NSF grant CCF-2114116.}\\
Northeastern University\\
ivanov.p@northeastern.edu
\and
Chin Ho Lee\thanks{Work done in part at Harvard University, supported by Madhu Sudan’s and Salil Vadhan’s Simons Investigator Awards.} \\
North Carolina State University\\
chinho.lee@ncsu.edu	
\and
Emanuele Viola\thanks{Supported by NSF grant CCF-2114116.} \\
Northeastern University\\
viola@ccs.neu.edu
}

\begin{document}

\maketitle

\begin{abstract}
We prove several new results on the Hamming weight of bounded uniform and small-bias distributions.

We exhibit bounded-uniform distributions whose weight is anti-concentrated, matching existing concentration inequalities.  This construction relies on a recent result in approximation theory due to Erd\'eyi (Acta Arithmetica 2016).  In particular, we match the classical tail bounds, generalizing a result by Bun and Steinke (RANDOM 2015).
Also, we improve on a construction by Benjamini, Gurel-Gurevich, and Peled (2012).

We give a generic transformation that converts any bounded uniform distribution to a small-bias distribution that almost preserves its weight distribution.  Applying this transformation in conjunction with the above results and others, we construct small-bias distributions with various weight restrictions.  In particular, we match the concentration that follows from that of bounded uniformity and the generic closeness of small-bias and bounded-uniform distributions, answering a question by Bun and Steinke (RANDOM 2015).  Moreover, these distributions are supported on only a constant number of Hamming weights.

We further extend the anti-concentration constructions to small-bias distributions perturbed with noise, a class that has received much attention recently in derandomization.  Our results imply (but are not implied by) a recent result of the authors (CCC 2024), and are based on different techniques.  In particular, we prove that the standard Gaussian distribution is far from any mixture of Gaussians with bounded variance.
\end{abstract}

\section{Introduction and our results}
A distribution $D$ over $\{-1,1\}^n$ is \emph{$(\eps,k)$-biased} if for every $S \subseteq [n]$ of size $0 < |S| \le k$, we have $|\mathbb{E}[D^S]| \le \eps$, where $D^S := \prod_{i\in S} D_i$.  If $\eps = 0$ then any $k$ bits are uniform and $D$ is called \emph{$k$-wise uniform} or simply \emph{$k$-uniform}.  If $k = n$ then $D$ is called \emph{$\eps$-biased}.

The study of these distributions permeates and precedes theoretical computer science.
They were studied already in the 40's under the name of orthogonal arrays \cite{MR22821}, and are closely related to universal hash functions \cite{CaW79}, error-correcting codes (see e.g.~\cite{DBLP:journals/eccc/HatamiH23}), and in their modern guise were introduced in the works \cite{ABI86,ChorGoHaFrRuSm85,NaN90}.

Exploited in countless works, one of the most useful properties of such distributions $D$ is that the distribution of their Hamming weight 
\[
  \summ  D := \sum_{i=1}^n D_i
\]
is approximately the (centered) binomial distribution
\[
  B := \summ U,
\]
where $U$ is uniform in $\{-1,1\}^n$.  (For simplicity, in this work we focus on the simplest setting where the distributions are supported on $\pmo^n$.)

Yet, perhaps surprisingly, available bounds were loose or only applied to specific settings of parameters.  In this work, we prove several new results on the Hamming weight of bounded-uniform distributions, small-bias distributions, and an extension of the latter obtained by perturbing it with noise.  We discuss each of these in turn.

\subsection{Bounded uniformity}
One of the key properties of bounded uniform distributions is their tail bounds.  They have myriad applications in a wide range of areas in computer science, including hashing, load balancing, streaming algorithms, derandomization and cryptography.  For example, the work \cite{SSS95} titled
``Chernoff--Hoeffding bounds for applications with limited independence" has 500+ references according to Google scholar.

These bounds can be obtained from applying Markov's inequality on an upper bound of the higher moments of the Hamming weight $\summ D := \sum_{i=1}^n D_i$ of the distribution~\cite{SSS95,BR94}.  See Section 3.4 in \cite{DuP05} for details of this proof strategy and others.

\begin{fact} \label{fact:tail-bound}
  Let $D$ be a $(2k)$-uniform distribution on $\pmo^n$.
  For every integer $t > 0$, we have
  \[
  \Pr\bigl[ \abs{\summ D} \ge t \bigr]
    \le \sqrt{2} \left(\frac{2kn}{e t^2}\right)^k .
  \]
\end{fact}

Note that one can replace $k$ with any $k' \le k$ in the bound, as any $k$-uniform distribution is also $k'$-uniform.
Thus, this tail bound becomes non-trivial whenever $t \ge c\sqrt{n}$.
In this paper, as in \cite{moti}, every occurrence of ``$c$" denotes a possibly different positive real number.  The notation ``$c_x$" for parameter(s) $x$ indicates that this number may depend on $x$ and only on $x$. Replacing ``$c$'' with $O(1)$ everywhere is consistent with one common interpretation of the big-O notation.

When $t \le c\sqrt{n \log k}$, pseudorandomness results on $k$-uniformity against halfspaces~\cite{benjamini2012kwise,DGJSV-bifh,DKN10} imply that every $k$-uniform distribution puts roughly the same mass on $\Pr[ \summ D \ge t]$ as $B$.

\begin{theorem}[\cite{benjamini2012kwise,DGJSV-bifh,DKN10}] \label{thm:halfspaces}
  Let $D$ be a $k$-uniform distribution on $\pmo^n$.
  For every integer $t > 0$, we have $\abs{\Pr[ \summ D \ge t ] - \Pr[ B \ge t ]} \le c/\sqrt{k}$.
\end{theorem}

Therefore, for thresholds in this regime, essentially the upper and lower bounds on the tail of the (centered) binomial distribution extend to $k$-uniform distributions.

We note that tight estimates on the moments of the sum of independent bounded random variables have been established~\cite{Skorski22}.  But these bounds do not imply the tightness of \Cref{fact:tail-bound}.
More generally, and perhaps surprisingly, lower bounds on the tail mass of $k$-uniform distributions remain scarce.
This question was revisited by Bun and Steinke~\cite{BS15}.
They show that for every $k \ge c\log n$, there exists a $k$-uniform distribution $D$ such that $\Pr[ \summ D \ge t] \ge c^{-k}$ for $t = c\sqrt{nk}$.
However, their result only applies to $k \ge c \log n$ and ties $t$ and $k$, and so it does not apply in several regimes of interest.

\paragraph{Our results.}
In this work, we obtain matching lower bounds to \Cref{fact:tail-bound,thm:halfspaces}, generalizing or strengthening a number of previous works.

\begin{theorem}\label{thm:k-wise-tail-lb-middle}
  For every $k$ and $t$, there exists a $k$-uniform distribution $D$ on $\pmo^n$ such that 
  \[
    \Pr[\summ D \ge t] - \Pr[B \ge t]
    \ge c \sqrt{\tfrac{n}{k}}\Pr[B = t]
    \ge c 2^{-t^2/n}/\sqrt{k}.
  \]
\end{theorem}

Note that $\Pr[B = t]$ on the right hand side cannot be replaced with $\Pr[B \ge t]$.

\Cref{thm:k-wise-tail-lb-middle} matches \Cref{thm:halfspaces} up to constant, in particular removing a logarithmic factor from a lower bound sketched in \cite{benjamini2012kwise}.
Moreover, it shows that \Cref{fact:tail-bound} is tight for $t \in [\sqrt{n},c\sqrt{nk}]$.

For $t \ge c\sqrt{nk}$, we obtain the following lower bound, again matching \Cref{fact:tail-bound}.

\begin{theorem} \label{thm:k-wise-tail-lb}
  For every $k \le (n/9)^{1/3}$ and $t \ge \sqrt{nk}$, there exists a $k$-uniform distribution $D$ on $\pmo^n$ such that $\Pr[ \summ D \ge t] \ge \frac{1}{3k^{3/2}} (\frac{kn}{16t^2})^{k/2}$.
\end{theorem}

\Cref{thm:k-wise-tail-lb} generalizes \cite{BS15}.
To illustrate this regime, note that in particular for $t = c\sqrt{n \log n}$ we show that the error is large: $\ge (c/\log n)^k$, whereas $\Pr[ B \ge t] \le 1/n^c$ is exponentially smaller.

There remain some parameter regimes that are not covered by our or previous works.
For example, does \emph{every} $k$-uniform distribution put mass at least $1/k^c$ on strings of weight $c \sqrt{n \log k}$?
In general, we raise the question of establishing the best possible error for any threshold $t$ between $k$-uniform distributions and binomial, interpolating between the error $c/\sqrt{k}$ for $t \le c\sqrt{n \log k}$~\cite{benjamini2012kwise,DGJSV-bifh,DKN10} and the exponentially small error for larger $t$ (\Cref{fact:tail-bound,thm:k-wise-tail-lb}).

\subsection{Small-bias}

We develop a paradigm to obtain small-bias distributions from $k$-uniform distributions while retaining some of their deviation properties.  The paradigm has two steps.  First, \emph{symmetrize} the distribution.  As $k$-uniform distributions are typically supported on nearly balanced strings, this step has the effect of making the bias small on tests of size not too large.  Second, add noise, following \cite{LV-sum}.  This makes the bias small on large tests.  Moreover, the small-bias distributions that we construct are supported on few Hamming weights.

\begin{lemma} \label{lem:bu-to-sb}
  Let $D$ be any $k$-uniform distribution.
  There exists a distribution $D_\bias$ supported on $ck^2$ Hamming weights that is simultaneously $k$-uniform and $(ck/n)^{k/4}$-biased such that for every interval $[a,b]$,
  \[
    \Pr\bigl[ \summ D_\bias \in [a-k, b+k] \bigr]
    \ge \Pr\bigl[\summ D \in [a,b] \bigr] .
  \]
\end{lemma}

We note that while the idea of adding noise is from \cite{LV-sum}, here we inject much less noise than in \cite{LV-sum}.  This is made possible by the symmetrization step, which lets us reduce the bias for moderate-size parities.  On the other hand, symmetrization destroys linearity which is essential in \cite{LV-sum}.

Using this approach, we carry results on the Hamming weight of $k$-uniform distributions, including our new results for $k$-uniformity (\Cref{thm:k-wise-tail-lb-middle,thm:k-wise-tail-lb}), to small-bias distributions.

First, we obtain the following theorem stating that there are small-bias distributions supported on few Hamming weights that are nearly ``balanced.''  This is obtained from the construction of $k$-uniform distributions in  \cite{BHLV} (\Cref{lem:BHLV}).

\begin{theorem}\label{thm:sb-concentrate}
  For every $k$, there exists a $(ck/n)^k$-biased distribution $D$ supported on at most $ck^2$ weights such that $\abs{\summ D} \le 21\sqrt{kn}$ always.
\end{theorem}

Now we state results corresponding to \Cref{thm:k-wise-tail-lb-middle} and \Cref{thm:k-wise-tail-lb}.

\begin{theorem}\label{thm:sb-tail-lb-middle}
  For every $k$ and $t$, there exists a $(ck/n)^k$-biased distribution $D$ supported on at most $ck^2$ weights such that 
  \[
    \Pr[\summ D \ge t-k] - \Pr[B \ge t]
    \ge c \sqrt{\tfrac{n}{k}}\Pr[B = t]
    \ge c 2^{-t^2/n}/\sqrt{k}.
  \]
\end{theorem}

In typical settings, one can replace $t-k$ with $t$.  However this in general depends on $t$ and $k$ so we leave the general statement for simplicity.

\begin{theorem} \label{thm:sb-anticoncentrate}
  For every $k \le (n/9)^{1/3}$ and $t \ge \sqrt{nk}$, there exists a $(ck/n)^{k/2}$-biased distribution $D$ supported on at most $ck^2$ weights such that
  \[
    \Pr[ \summ D \ge t ] \ge \frac{1}{3k^{3/2}} \Bigl( \frac{kn}{64t^2} \Bigr)^{k/2} .
  \]
\end{theorem}

Several researchers, see for example \cite{BS15}, posed the question of understanding tail bounds for small-bias distributions.
And yet, our understanding of their tail probabilities is notably lacking.
In terms of upper bounds, existing tail bounds follow from \Cref{fact:tail-bound} via a connection to the closeness between small-bias and $k$-uniform distributions~\cite{AGM03,AlonAKMRX07,OZ18}. 
On the opposite direction, Bazzi~\cite{Bazzi15} showed that an exponentially small-biased distribution does not fool the majority function with error $c/\sqrt{n}$.
The recent work \cite{SSS-I} shows that $\Pr[ \summ D \ge \sqrt{nk}] \ge c^k$ for some $(ck/n)^k$-biased distribution $D$.  However, this result applies to a specific threshold and does not scale as the threshold varies.

By contrast, \Cref{thm:sb-anticoncentrate} shows that the tail bounds for bounded uniformity cannot be improved even for small-bias distributions, in general parameter settings.  This in particular answers a question in \cite{BS15} and generalizes the recent work \cite{SSS-I}.

\subsection{Small-bias plus noise}
In the past decade, researchers have considered bounded uniform and small-bias distributions which are perturbed by noise.

\begin{definition}
$N_\rho$ is the noise distribution on $\pmo^n$, where each bit is independently set to uniform with probability $1-\rho$ and $1$ otherwise.  We write $D \cdot N_\rho$ for the coordinate-wise product of $D$ and $N_\rho$, which corresponds to bit-wise xor over $\{0,1\}$.  We also call this the \emph{$\rho$-smoothed distribution}. Note that $x \cdot N_1 = x$ and $x \cdot N_0 = U$, for any $x$.
\end{definition}

Note that smoothing does not increase the distance of any two distributions, with respect to any class of tests which is closed under shifts.  So distinguishing smoothed distributions is at least as hard as distinguishing the corresponding (non-smooth) distributions.
A main motivation for considering smoothed distributions comes from several paradigms for constructing pseudorandom generators (PRGs) that combine $(\eps,k)$-biased distributions in different
ways.  These paradigms have been proposed in the last 15 years or so; for additional background, we refer the readers to the recent monograph \cite{DBLP:journals/eccc/HatamiH23}.

A main result from \cite{SSS-I} is that smoothed $n^{-k}$-bias distributions do not fool thresholds with error less than $c^k$.
This is then used to show that they do not fool, even with constant error, other models such as small-space algorithms or small-depth circuits.

\paragraph{Our results.}
Using the small-bias distributions constructed in this work, we obtain alternative proofs of the main result from \cite{SSS-I}.
We note that these proofs provide different information.
In \cite{SSS-I}, the pseudorandom distributions put \emph{more} mass than the binomial on the tail.
In some of the proofs presented here, the mass will be \emph{less} and the distribution is supported only on a \emph{few} weights.
This also provides a more complete picture of how these distributions can be designed.

In more detail, the small-bias distributions given in \Cref{thm:sb-anticoncentrate,thm:sb-concentrate} put noticeable different masses on its tail than the binomial distribution.
From this, one can argue that they remain so after perturbed with noise, and thus can be distinguished by thresholds.

\begin{theorem} \label{thm:threshold-ck-distinguish}
For any $\rho\in(0,1]$ and $k\leq c \rho^2 n^{1/3}$, there is a $(ck/n)^k$-biased distribution $D$ on $\pmo^n$ and some threshold $t = c\sqrt{nk}/\rho$ such that
\[
  \Pr [B \ge t ]
  \ge \Pr [ \summ{D \cdot N_\rho} \ge t ] + 2^{-ck/\rho^2}.
\]
\end{theorem}

Moreover, we show that a threshold can distinguish smoothed small-bias from \emph{any} $k$-uniform distribution, answering a question that was left open in \cite{SSS-I} in the affirmative.

\begin{theorem} \label{thm:anticoncentration-vs-threshold}
  For every $k \le (n/9)^{1/3}$ and $\rho \in [0,1)$, let $k' := c\log(1/\rho) k$.
  There exists a $(ck/n)^{k/2}$-biased distribution $D$ and a threshold $t$ such that for every $k'$-uniform distribution $D_{k'}$,
  \[
    \Pr[ \summ{D \cdot N_\rho} \ge t ] 
    \ge \Pr[ \summ D_{k'} \ge t] + \Bigl(\frac{c\rho^2}{\log(1/\rho)}\Bigr)^{k/2} .
  \]
\end{theorem}
Note that \Cref{thm:anticoncentration-vs-threshold} shows that a smoothed small-bias distribution puts more mass on the tail than the uniform distribution (which is also $k'$-uniform), whereas \Cref{thm:threshold-ck-distinguish} shows the opposite.

In addition, we show that any distribution supported on a few weights, after perturbed with noise, can be distinguished from the binomial distribution by a threshold.
This uses a new bound on the total variation distance between any mixture of few Gaussian distributions with bounded variance the standard Gaussian distribution (\Cref{lem:mixture-distance-improved}).

\begin{theorem} \label{thm:threshold-caratheodory}
  Let $D$ be any distribution on $\pmo^n$ supported on $k$ weights.
  For any $\rho\in(0,1]$, there is some $t$ such that
  \[
    \abs[\big]{ \Pr [B \ge t ] -\Pr [ \summ{D \cdot N_\rho} \ge t ] }
    \ge   2^{-ck/\rho}.
  \]
\end{theorem}

Applying \Cref{thm:threshold-caratheodory} to our small-bias distributions $D$ that are supported on $ck^2$ weights implies that their smoothed distributions $D \cdot N_\rho$ can be distinguished from uniform with advantange $2^{-ck^2/\rho}$.
We remark that one can further improve the advantage to $2^{-ck/\rho}$, by applying \Cref{lem:cara-kwise} to ``sparsify'' any $k$-uniform distribution so that it is supported on $k+1$ weights, and then observing that $D \cdot N_{2\rho} \equiv (D \cdot N_\rho) \cdot N_\rho$ and $D \cdot N_\rho$ is small-biased.

\section{Proof of \Cref{thm:k-wise-tail-lb}}

At a high level, the proof of \Cref{thm:k-wise-tail-lb} follows the same strategy in \cite{BS15}.
However, there are some noticeable differences.
First, we decouple the threshold and the error parameters.
Second, we do not pass the argument to Gaussian distributions.
Finally and most importantly, their proof incurs a loss of a $1/\sqrt{n}$ factor in their lower bound on the tail mass (and thus requires $k \ge c\log n$), which is significant in certain regimes of parameters.

To prove \Cref{thm:k-wise-tail-lb}, we use tools in approximation theory.
In particular, to remove the $1/\sqrt{n}$ loss in \cite{BS15}, instead of applying a Markov--Bernstein type inequality for $L_\infty$-norms, we rely on the following inequality by Erd\'eyi,

\begin{lemma}[Theorem 2.1 ($q=1$ case) in \cite{Erd16}] \label{lem:erdelyi}
  For $m \in \N$ and $L > 0$, let $Q \in \C[x]$ be a degree-$d$ univariate polynomial (possibly with complex coefficients) such that
  \[
    \abs{Q(0)}
    > \frac{1}{L} \Bigl(\sum_{j=1}^m \abs{Q(j)} \Bigr).
  \]
  Then $d \ge 7\sqrt{m/L}$.
\end{lemma}

We also need the following inequality due to Ehlich, Zeller, Coppersmith, Rivlin, and Cheney.

\begin{lemma}[Lemma 20 in \cite{BS15}] \label{lem:coppersmith-rivlin}
  Let $p$ be a univariate degree-$d$ polynomial such that $\abs{p(i)} \le 1$ on $i \in \{0, \ldots, m\}$, where $3d^2 \le m$.
  Then $\abs{p(x)} \le 3/2$ for every $x \in [0,m]$.
\end{lemma}

We will also use the following extremal property of Chebyshev polynomials $T_k$ (cf. \cite[Propositions~2.4 and 2.5]{SV13-book}).

\begin{fact} \label{fact:chebyshev}
  Let $p$ be a univariate polynomial of degree $k$ such that $\abs{p(t)} \le 1$ on $[-1,1]$.
  For every $s \ge 1$, $\abs{p(s)} \le T_k(s) \le (2\abs{s})^k$, where $T_k$ is the Chebyshev polynomial of degree $k$.
\end{fact}

\begin{fact}[cf. Lemma~23 in \cite{BHLV}] \label{fact:stirling}
  If $a$ is an integer such that $\abs{a} \le n$ and $a \equiv n \bmod 2$, then $\Pr[B = a] \ge 2^{-a^2/n} \frac{1}{2\sqrt{n}}$.
\end{fact}

\paragraph{Proof of \Cref{thm:k-wise-tail-lb}.}
By strong duality (cf. \cite{BS15}), we have
  \[
    \max_D \Pr\bigl[\Id(\summ D \ge t)\bigr] = \min_p \E\bigl[p(U)\bigr],
  \]
  where the maximum is over all $k$-uniform distributions $D$, and the minimum is over all degree-$k$ (upper sandwiching) polynomials $p\colon\pmo^n \to \R$ such that $p(x) \ge \Id(\summ x \ge t)$.

  Let $p$ be a degree-$k$ polynomial attaining $\delta := \min_p \E[p(U)]$.
  Define the univariate polynomial $q$ to be the symmetrization of $p$, that is, $q(\sum_{i=1}^n x_i) := p(x)$.

  We use \Cref{lem:erdelyi,lem:coppersmith-rivlin} to bound $q(t)$ on $t \in [-\sqrt{kn},\sqrt{kn}]$.
  We first state the lemma and defer its proof to the end of this section.
  \begin{lemma} \label{lem:sup-bound-interval}
    Suppose $3k^2 \le \sqrt{kn}$.
    Then we have $\abs{q(t)} \le 3\delta \cdot k^{3/2} \cdot 2^{k}$ for every $t \in [-\sqrt{kn}, \sqrt{kn}]$.
  \end{lemma}
  Note that the upper bound is independent on $n$, whereas the bound in \cite[Theorem 9]{BS15} has a polynomial dependence on $n$.

  We now continue the proof assuming \Cref{lem:sup-bound-interval}.
  Let $s := \sqrt{kn}$.
  Observe that $q(t) \ge \Id(t \ge t) = 1$.
  Let $\wt{q}_s(\theta) = q(\theta s)$.
  By \Cref{lem:sup-bound-interval}, we have $\max_{\theta \in [-1,1]} \abs{\wt{q}(\theta)} \le 3 \delta \cdot k^{3/2} \cdot 2^{k}$.
  By \Cref{fact:chebyshev}, for $t \ge s$, we have
  \begin{align*}
    1
    \le q(t)
    &= \wt{q}_s\Bigl(\frac{t}{s}\Bigr) \\
    &\le \Bigl(\frac{2t}{s}\Bigr)^k \cdot \max_{\theta \in [-1,1]} \abs{\wt{q}_s(\theta)} \\
    &\le \Bigl(\frac{2t}{\sqrt{kn}}\Bigr)^k \cdot 3 \delta \cdot k^{3/2} \cdot 2^{k} .
  \end{align*}
  Rearranging gives $\delta \ge \frac{1}{3k^{3/2}} (\frac{kn}{16t^2})^{k/2}$, proving \Cref{thm:k-wise-tail-lb}. \qed

  \medskip

  To prove \Cref{lem:sup-bound-interval}, we use \Cref{lem:erdelyi} to bound $q(t)$ on the integer points between $-\sqrt{kn}$ and $\sqrt{kn}$, and then extend the bound to the whole interval using \Cref{lem:coppersmith-rivlin}.

  \begin{claim} \label{claim:sup-bound-integers}
    $0 \le q(t) \le 2\delta \cdot k^{3/2} \cdot 2^k$ for every $t \in \{-\sqrt{kn}, \ldots, \sqrt{kn} \}$, 
  \end{claim}

  \begin{proof}[Proof of \Cref{claim:sup-bound-integers}]
    Assume $n$ is even so that $B$ is supported on even integers.
    Fix an even integer $t_0$ such that $\abs{t_0} \le \sqrt{kn}$ and without loss of generality assume $t_0 > 0$. 
    As $\Pr[B = t] \ge \Pr[B = \sqrt{nk}]$ for every even $t$ with $\abs{t} \le \sqrt{nk}$ and $q$ is nonnegative, we have
    \begin{align*}
      \Pr\Bigl[B = \sqrt{kn}\Bigr] \sum_{j=1}^{\sqrt{nk}/2} q(t_0 - 2j)
      &\le \sum_{j=1}^{\sqrt{nk}/2} \Pr\bigl[ B = t_0 - 2j \bigr] q(t_0 - 2j) \\
      &\le \E[q(B)]
      = \delta .
    \end{align*}
    Rearranging and using \Cref{fact:stirling} gives
    \[
      \sum_{j=1}^{\sqrt{nk}/2} q(t_0 - 2j)
      \le \frac{\delta}{\Pr[B = \sqrt{kn}]}
      \le 2\delta \cdot \sqrt{n} \cdot 2^k .
    \]
    Consider the polynomial $Q(t) := q(t_0 - 2t)$ of degree $k$.
    Let $m = \sqrt{nk}/2$, and $L = m/k^2 = n^{1/2} k^{-3/2}$.
    As $k < 7k = 7\sqrt{m/L}$, by the contrapositive of \Cref{lem:erdelyi}, we have
    \begin{align*}
      \abs{q(t_0)}
      = \abs{Q(0)}
      &\le \frac{1}{L} \sum_{j=1}^m \abs{Q(j)} && \text{(\Cref{lem:erdelyi})} \\
      &= \frac{1}{L} \sum_{j=1}^{\sqrt{kn}/2} q(t_0 - 2j) \\
      &\le 2\delta \cdot \frac{\sqrt{n} \cdot 2^k}{L}
      = 2\delta \cdot k^{3/2} \cdot 2^k . \qedhere
    \end{align*}
  \end{proof}
  \begin{proof}[Proof of \Cref{lem:sup-bound-interval}]
    Let $Q$ be the degree-$k$ polynomial $Q(j) := q(\sqrt{kn} - 2j)$.
    By \Cref{claim:sup-bound-integers}, we have $\abs{Q(j)} \le M$ on $j \in \{0, \ldots, \sqrt{kn}\}$, where $M := 2\delta \cdot k^{3/2} \cdot 2^{k}$.
    As $3k^2 \le \sqrt{kn}$ for $k \le (n/9)^{1/3}$, by \Cref{lem:coppersmith-rivlin}, we have $q(t) \le 3M/2$ for every $t \in [0,\sqrt{kn}]$.
  \end{proof}

\section{Proof of \Cref{thm:k-wise-tail-lb-middle}} \label{sec:k-wise-tail-middle}

We rely on the result from \cite[Theorem 2]{BHLV} that for every $a$, $m$, and $k \le n/(8m^2)$, there is a $k$-uniform distribution $D$ supported on $\{x \in \pmo^n : \summ x \equiv a \pmod{m} \}$.
Moreover, implicit in the proof they show that the probability mass on every point $s$ in the support of $D$ is at least $(m/4) \cdot \Pr[B = s]$.

Let $m := \sqrt{n/(8k)}$.
We can pick an integer $a$ such that $t$ belongs to the support of some $k$-uniform $D$, from which we conclude that 
\[
  \eps := \Pr[ \summ D = t] - \Pr[ B = t]
  \ge (m/4-1) \Pr[B = t] .
\]
Now, we have either
\begin{align*}
  \Pr[ \summ D \ge t] - \Pr[ B \ge t]
  &\ge \eps/2 \text{ or } \\
  \Pr[ \summ D \le t] - \Pr[ B \le t]
  &\ge \eps/2 ,
\end{align*}
as otherwise, summing both inequalities give $(1 + \Pr[\summ D = t]) - (1 + \Pr[ B = t]) < \eps$, a contradiction.
If we are in the second case, we can consider $\ol{D}$, the complement of $D$, which is also $k$-uniform, and we have $\Pr[\summ D \le t] = \Pr[\summ {\ol{D}} \ge t]$.

\section{From bounded-uniformity to small-bias}
In this section, we prove \Cref{lem:bu-to-sb}, which gives a generic way to transform any $k$-uniform distribution into a $(ck/n)^{k/4}$-bias distribution while preserving some of the weight properties of the distribution.

We start with a lemma that lets us ``sparsify'' the weight distribution of any $k$-uniform distribution.
This uses Carath\'eodory's theorem from convex geometry, stated next, which has a simple proof (see e.g.\ Theorem~2.3 in Chapter~1 in \cite{MR1940576}).

\begin{lemma}[Carath\'eodory's theorem] \label{lem:carath}
    Every point in the convex hull of a set $S \subseteq \R^k$ can be represented as a convex combination of $k+1$ points from $S$.
\end{lemma}

We use \Cref{lem:carath} to sparsify any \emph{symmetric} $k$-uniform distribution so that it is supported on $k+1$ weights.
Here we use the fact that a symmetric distribution on $\pmo^n$ is $k$-uniform if and only if the first $k$ moments of $\summ D$ match the ones of $B$.

\begin{lemma} \label{lem:cara-kwise}
  Let $D$ be any symmetric $k$-uniform distribution such that $\summ D$ is supported on a set $S \subseteq \pmo^n$.
  Then there is a symmetric $k$-uniform distribution $D'$ such that $\summ D'$ is supported on a subset $S' \subseteq S$ of size at most $k+1$.
\end{lemma}
\begin{proof}
  For each $i$, define $p_i := \Pr[ \summ D = i]$ and $v^i := (i, i^2, \ldots, i^k) \in \R^k$.
  Let $b \in \R^k$ be the vector $(\E[B], \E[B^2], \ldots, \E[B^k])$.
  As $D$ is $k$-uniform, we have $b = \sum_{i \in S} p_i v^i$, and so $b$ lies in the convex hull of the $v_i$'s.
  Thus, there is a set $S'$ of $k+1$ indices such that $b = \sum_{i \in S'} q_i v_i$.

  We define the symmetric distribution $D'$ by $\Pr[ \summ D = i] := q_i \Id(i \in S')$
  It is clear that $\sum_i \Pr[ \summ D' = i] v^i = b$, and thus $D'$ is $k$-uniform.
\end{proof}

We will also use the following bound on the bias of the uniform distribution on a ``Hamming slice.''
This follows from an upper bound on Krawtchouk polynomials that was established in \cite{SSS-I}.

\begin{lemma}[Claim 31 and Corollary 16 in \cite{SSS-I}] \label{lem:bias-unif-weight}
  Let $W_t$ be the uniform distribution on $\{x \in \pmo^n: \sum_{i=1}^n x_i = t \}$.
  For every subset $S \subseteq [n]$ of size $\ell$, we have
  \[
    \abs[\big]{ \E\bigl[ W_t^{[n]\setminus S} \bigr] }
    = \abs[\big]{ \E\bigl[ W_t^S \bigr] }
    \le \Bigl( \frac{\ell}{n} + \frac{t^2}{n^2} \Bigr)^{\frac{\ell}{2}} .
  \]
\end{lemma}

\paragraph{Proof of \Cref{lem:bu-to-sb}.}
  Let $D$ be any $k$-uniform distribution on $\pmo^n$.
  We obtain our small-bias distribution $D_\bias$ in 3 steps as follows:
  
  First, we symmetrize $D$ to obtain $D_\sym$.
  Then, we use \Cref{lem:cara-kwise} to sparsify $D_\sym$ to obtain $D'$ so that it is supported on $k+1$ weights.
  Finally, we add $k/2$ bits of noise, obtained by repeating the following process independently $k/2$ times:
  pick a uniform random coordinate of $D'$ and set it to uniform.
  
  \medskip
  
  We now claim that $D_\bias$ is $(ck/n)^{\frac{k}{4}}$-biased.
  Let $S \subseteq [n]$ be any non-empty subset.
  Note that the noise applied in the last step can only reduce the bias of $D'$.
  We consider three cases:
  \begin{enumerate}
      \item $\abs{S} \in [1,k]$: observe that $D'$ remains $k$-uniform.
        So $\E[D'^S] = 0$.
      \item
      $\abs{S} \in [k+1, n-k-1]$: let $t = (kn^3)^{\frac{1}{4}}$.
        By \Cref{lem:bias-unif-weight,fact:tail-bound}, we have 
      \begin{align*}
      \abs[\big]{\E[D'^S]}
        &\le \abs[\big]{\E\bigl[D'^S \mid \abs{\summ D'} \le t \bigr]} + \Pr\bigl[ \abs{\summ D'} > t \bigr] \\
        &\le \left( \frac{\abs{S}}{n} + \frac{t^2}{n^2} \right)^{\abs{S}/2} +  \sqrt{2} \left( \frac{kn}{et^2} \right)^{k/2} \\
        &\le 2 \cdot (2k/n)^{\frac{k}{4}} .
      \end{align*}
      \item 
        $\abs{S} \in [n-k, n]$: If a bit in $S$ is set to uniform by the noise, which happens with probability $(1 - k/n)$, then the bias is 0.
        So
      \[
        \abs[\big]{\E[D_\bias^S]}
        \le (1 - (1 - k/n))^{k/2}
        = (k/n)^{k/2} .
      \]
  \end{enumerate}
  Since changing a bit $x_i \in \pmo$ of $x$ can only change $\summ x$ by at most $2$, the lemma follows.  \qed

\subsection{Proofs of \Cref{thm:sb-concentrate,,thm:sb-tail-lb-middle,thm:sb-anticoncentrate}}

\Cref{thm:sb-tail-lb-middle,thm:sb-anticoncentrate} directly follow from applying \Cref{lem:bu-to-sb} to \Cref{thm:k-wise-tail-lb-middle,thm:k-wise-tail-lb} respectively.
\Cref{thm:sb-concentrate} follows from applying the lemma to following result, which exhibits  bounded uniform distributions that are supported only on nearly-balanced strings.

\begin{lemma}[\cite{BHLV}] \label{lem:BHLV}
  For any  integer $k$, there is a distribution $D$ supported on $\{ x\in \pmo^n: \abs{\summ{x}} \le 10\sqrt{kn}\}$ which is $(2k)$-uniform.
\end{lemma}

We note that one can instead apply \Cref{lem:bu-to-sb} to the standard randomness-efficient construction of $k$-uniform distributions via BCH codes, and then use results from algebraic geometry to bound the Hamming weight, specifically Theorem 18 in \cite{MR465510}.
However, the support is slightly less concentrated to the center than \Cref{lem:BHLV}.

\begin{claim} 
  There exists a linear distribution $D$ supported on $\{x \in \pmo^n: \abs{\summ{x}} \le c k\sqrt{n} \}$ which is $k$-uniform.
\end{claim}

\section{Distinguishing small-bias plus noise} \label{sec-alt-proofs}

In this section we prove \Cref{thm:threshold-ck-distinguish,,thm:anticoncentration-vs-threshold,thm:threshold-caratheodory}.
We will apply several results on sums of independent random variables.
We first state these results before proceeding to the proofs.

Let $\calN(0,1)$ denote the standard normal distribution, which has mean $0$ and variance $1$.

\begin{lemma}[Theorem~11.2 in~\cite{DasGupta-book}] \label{lem:Berry-Esseen}
Let $Y_1, \ldots, Y_n$ be $n$ independent random variables with $\E[Y_i]=0$, $\Var[Y_i] = \sigma_i^2$, $\E[\abs{Y_i}^3] < \infty$.
Let $Y := \sum_{i=1}^n Y_i$.
For every $\theta \in \R$,
\[
  \abs[\bigg]{ \Pr\biggl[ \frac{Y}{ \bigl(\sum_{i=1}^n \sigma_i^2\bigr)^{1/2} } \ge \theta \biggr] - \Pr\biggl[ \calN(0,1) \ge \theta \biggr] }
  \le \frac{\sum_{i=1}^n \E[\abs{Y_i}^3]}{\bigl(\sum_{i=1}^n \sigma_i^2\bigr)^{3/2}} .
\]
\end{lemma}

For fixed $\rho_i$'s and $\sigma_i$'s, the additive error given by \Cref{lem:Berry-Esseen} is roughly $1/\sqrt{n}$.
So, for \Cref{thm:threshold-caratheodory} to hold with $k \ge c\log n$, we would need a more refined approximation.
For this reason,  we will be using the following Cram\'er's estimate of sums of independent random variables, which gives a multiplicative rather than additive approximation in terms of $\calN(0,1)$.

\begin{lemma}[Chapter VIII, Equation~(2.41) in \cite{Petrov-book}] \label{lem:Petrov}   There exists a constant $c > 0$ such that the following holds.
  Let $Y_1,\ldots, Y_n$ be $n$ independent random variables with $\E[Y_i]=0$, and $\E[Y_i^2] = \sigma_i^2$ for each $i \in [n]$.
  Let $Y := \sum_{i=1}^n Y_i$.
  For $0\le \theta \le cn^{1/6}$,
  there exists an $\eps \in [0,\frac{(\theta+1)}{c\sqrt{n}}]$ such that
  \[
    \Pr\biggl[\frac{Y}{{\bigl(\sum_{i=1}^n \sigma_i^2\bigr)^{1/2}}}\geq\theta\biggr]
    = \Pr\Bigl[ \calN(0,1) \ge \theta \Bigr] \cdot \exp\biggl(\frac{\sum_{i=1}^n\E[Y_{i}^{3}]}{6 \bigl(\sum_{i=1}^n \sigma_i^2\bigr)^{3/2}} \cdot\theta^3 \biggr) (1 + \eps) .
  \]
\end{lemma}

To relate \Cref{lem:Petrov} to \Cref{lem:Berry-Esseen}, note that when $\theta$ is small, $\exp(\frac{\sum_{i=1}^n\E[Y_{i}^{3}]}{6 (\sum_{i=1}^n \sigma_i^2 )^{3/2}} \cdot\theta^3)$ is roughly $1 +  \frac{\sum_{i=1}^n\E[Y_{i}^{3}]}{(\sum_{i=1}^n \sigma_i^2)^{3/2}}$.

Specializing \Cref{lem:Petrov} to our applications, we obtain the following lemma.
We first need a simple claim.

\begin{claim} \label{claim:noise-moments}
  For every $x \in \pmo^n$ and $i \in [n]$, let $Y_i$ be the mean zero variables $Y_i := (x \cdot N_{\rho})_i - \rho x_i$.
  Then $\E[Y_i^2] = 1-\rho^2$ and $\E[Y_i^3] \le -2\rho(1-\rho^2) x_i \in [-1,0]$.
\end{claim}
\begin{proof}
  It suffices to compute the first 3 moments of $(x \cdot N_\rho)_i$.
  We have $\E[(x \cdot N_{\rho})_i] = \rho x_i$, $\E[(x \cdot N_{\rho})_i^2] = 1$, and $\E[(x \cdot N_{\rho})_i^3] = \rho x_i$.
\end{proof}

\begin{lemma}\label{lem:petrov-noise}
  For every $x \in \pmo^n$, $\rho \in [0,1)$, and $\theta \in [0, cn^{1/6}]$, we have
  \begin{enumerate}
    \item $\Pr[B \ge \sqrt{n} \cdot \theta] \ge \Pr[ \calN(0,1) \ge \theta]$, and
    \item $\Pr\Bigl[ \frac{\summ{x \cdot N_\rho} - \rho \cdot \summ x}{\sqrt{n} (1-\rho^2)^{1/2}} \ge \theta \Bigr] \le 2 \Pr\Bigl[ \calN(0,1) \ge \theta\Bigr]$.
  \end{enumerate}
\end{lemma}

\begin{proof}
  We apply \Cref{lem:Petrov} to obtain both inequalities.
  For the first inequality, observe that $\E[Y_i^3] = \E[B_i^3] = 0$.
  For the second inequality, we apply \Cref{lem:Petrov} using \Cref{claim:noise-moments}, which gives $\E[Y_i^3] \le 0$. 
  Note that $e^y \le 1$ for any $y \le 0$, and $c(1+\theta)/\sqrt{n} \le 1$.
\end{proof}

We will use the following approximation to compare the tails of the standard normal distribution.

\begin{lemma}[Lemma~22.2 in \cite{MR4201399}] \label{lem:bounds-phi-tail}
For any $\theta > 0$, 
\[
    \frac{1}{\theta+\frac{1}{\theta}}\leq \mathbb{P}\bigl[\calN(0,1)\ge \theta\bigr] \cdot \frac{\sqrt{2\pi}}{e^{-\theta^{2}/2}}\leq\frac{1}{\theta}.
\]
\end{lemma}

Finally, we state the following tail bounds for sums of independent bounded random variables with small variances.

\begin{claim}[Bernstein's inequality] \label{claim:bernstein}
  Let $Y_1, \ldots, Y_n$ be independent mean-zero random variables.
  Suppose $\abs{Y_i} \le M$ for every $i \in [n]$.
  Then
  \[
    \Pr\Bigl[ \sum_{i=1}^n Y_i \ge t\Bigr]
    \le \exp\Bigl( - \frac{t^2/2}{\sum_{i=1}^n \E[Y_i^2] + Mt/3} \Bigr) .
  \]
\end{claim}

Again, specializing \Cref{claim:bernstein} to our applications, we obtain the following.

\begin{claim} \label{claim:noise-tail}
  $\Pr[ \abs{\summ(x \cdot N_\rho) - \rho \cdot \summ x} \ge s] \le 2 \exp( -\frac{c s^2}{(1-\rho^2)n + s})$ for every $x \in \pmo^n$ and $s > 0$.
\end{claim}
\begin{proof}
  We will only prove that $\Pr[ \summ (x \cdot N_\rho) - \rho \cdot \summ x \ge s ] \le \exp(- \frac{c s^2}{ (1-\rho^2) n + s} )$.
  The other direction is analogous and the conclusion follows from a simple union bound.
  Let us consider the mean zero variables $Y_i := (x \cdot N_\rho)_i - \rho x_i$.
  We have $\abs{Y_i} \le 1+\rho \le 2$.
  Applying \Cref{claim:bernstein} with \Cref{claim:noise-moments}, we have
  \begin{align*}
    \Pr\bigl[ \summ (x \cdot N_\rho) - \rho \cdot \summ x \ge s \bigr]
    = \Pr\Bigl[ \sum_{i=1}^n Y_i \ge s \Bigr]
    &\le \exp\Bigl(- \frac{c s^2}{ (1-\rho^2) n + s} \Bigr) . \qedhere
  \end{align*}
\end{proof}

\subsection{Proof of \Cref{thm:threshold-ck-distinguish}}

Let $D$ be the $(ck/n)^k$-biased distribution in \Cref{thm:sb-anticoncentrate} such that $\abs{\summ D} \le 21\sqrt{kn}$.
Let $t := \beta \sqrt{kn}/\rho$ for a sufficiently large constant $\beta > 0$, and $\theta = t/\sqrt{n} = \beta \sqrt{k}/\rho$.
Using $(1-\rho^2)^{-1/2} \ge 1 + \rho^2/2$, we have
\begin{align} \label{eq:normalization}
  \frac{t - 21 \rho \sqrt{kn}}{\sqrt{n} (1-\rho^2)^{1/2}}
  = \sqrt{k} \cdot \frac{\tfrac{\beta}{\rho} - 21 \rho}{(1-\rho^2)^{1/2}}
  \ge  \beta \frac{\sqrt{k}}{\rho} \cdot \Bigl(1 - \frac{21\rho^2}{\beta}\Bigr) \cdot \Bigl(1 + \frac{\rho^2}{2}\Bigr) 
  \ge  (1 + \rho^2/4) \cdot \theta .
\end{align}
By assumption, $k \le c\rho^2 n^{1/3}$ and so $\theta \le cn^{1/6}$ for which \Cref{lem:petrov-noise} applies.
As $\abs{\summ D} \le 21 \sqrt{kn}$, using the second inequality in \Cref{lem:petrov-noise} and \cref{eq:normalization}, we have
\begin{align} 
  \Pr\bigl[ \summ (D \cdot N_\rho) \ge t \bigr]
  &\le \Pr\Bigl[ \summ (D \cdot N_\rho) - \rho \cdot \summ D  \ge t - 21 \rho \sqrt{kn}  \Bigr] \nonumber \\
  & \le  2 \Pr \bigl[ \mathcal{N}(0,1) \geq (1+\rho^2/4) \cdot \theta  \bigr] . \label{eq:petrov-smpn}
\end{align}
On the other hand, from the first inequality of \Cref{lem:petrov-noise}, we have
\begin{align} \label{eq:petrov-B}
  \Pr\bigl[ B \ge t \bigr]
  \ge \Pr \bigl[ \mathcal{N}(0,1)  \ge  \theta \bigr]. 
\end{align}
By \Cref{lem:bounds-phi-tail},
\begin{align}
  \Pr[ \calN(0,1) \ge \theta ]  \nonumber
  &\ge \frac{1}{2\theta} \frac{1}{\sqrt{2\pi}} e^{-\theta^2/2} \nonumber \\
  &\ge  \frac{c\sqrt{k}}{\rho}e^{-\theta^2/2} + \frac{2}{\theta} \frac{1}{\sqrt{2\pi}} e^{-(1+\rho^2/4)\theta^2/2} \nonumber \\
  &\ge e^{-ck/\rho^2} + 2\Pr\bigl[ \calN(0,1) \ge (1+\rho^2/4) \cdot \theta \bigr] . \label{eq:compare-gaussian}
\end{align}
Putting \cref{eq:petrov-smpn,,eq:petrov-B,eq:compare-gaussian} together completes the proof. \qed

\subsection{Proof of \Cref{thm:anticoncentration-vs-threshold}}

Let $k' = C\log(1/\rho)k$ and $t = \sqrt{k'n}$.
For every $k'$-uniform distribution $D_{k'}$, by \Cref{fact:tail-bound}, we have
\[
  \Pr[ \summ D_{k'} \ge t ]
  \le \sqrt{2} \Bigl( \frac{k'n}{et^2} \Bigr)^{k'/2}
  \le \rho^{Ck/2} .
\]
We will construct a small-distribution $D$ that puts more mass on the tail of an even larger threshold than $t$.
Specifically, let $t' = 2t/\rho$.
Applying \Cref{thm:sb-anticoncentrate} to $k$ and $t'$, we obtain a $(ck/n)^{k/2}$-biased distribution $D$ such that
\[
  \Pr[ \summ D \ge t']
  \ge \frac{1}{3k^{3/2}} \Bigl(\frac{\rho^2}{256C \log(1/\rho)}\Bigr)^{k/2} .
\]
We now show that conditioned on $\summ D \ge t'$, the smoothed distribution $D \cdot N_\rho$ still puts at least half the mass beyond $t$.
Since $\rho t'- t \ge t$, by \Cref{claim:noise-tail},
\[
  \Pr\bigl[ \abs{\summ(x \cdot N_\rho) - \rho \cdot \summ x} \ge t\bigr]
  \le 2 e^{-\frac{c t^2}{(1-\rho^2)n + t}}
  \le 1/2 .
\]
Therefore
\[
  \Pr[ \summ (D \cdot N_\rho) \ge t ] 
  \ge \frac{1}{6k^{3/2}} \Bigl(\frac{\rho^2}{256C\log(1/\rho)}\Bigr)^{k/2}
  \ge \Bigl(\frac{c\rho^2}{\log(1/\rho)} \Bigr)^{k/2} . \qed
\]

\subsection{Proof of \Cref{thm:threshold-caratheodory}}

We first state the main technical result we need, which may be of independent interest.
We defer its proof to the next section.
\begin{lemma}\label{lem:mixture-distance-improved}
  Let $M$ be a mixture of $k$ Gaussian distributions each with variance $\sigma^2 = 1-\rho^2$. Then there exists an interval $I$ such that
  \[
    \abs[\big]{ \Pr [\calN(0,1) \in I]- \Pr [M\in I] } \geq 2^{-c k/\rho }.
  \]
  In particular, up to a factor 2 the same bound applies with the interval replaced with some threshold.%Moreover, this threshold is at most $c\sqrt{k/\rho}$.
\end{lemma}

\begin{proof}[Proof of \Cref{thm:threshold-caratheodory} assuming \Cref{lem:mixture-distance-improved}]
  We first claim that $D \cdot N_\rho$ is $(c n^{-1/2})$-close to a mixture of $k$ Gaussian distributions with variance $1-\rho^2$ in Kolmogorov (aka. CDF) distance.
  This follows from applying \Cref{lem:Berry-Esseen,claim:noise-moments} to any fixed weight $w$ of $D$.
  Specifically, we have that for every $x \in \pmo^n$ and $\theta$,
  \begin{align*}
    \abs[\big]{\Pr\bigl[ \summ (x \cdot N_\rho) \ge \theta \sqrt{n} \bigr] -  \Pr \bigl[ \mathcal{N}(\mu,1-\rho^2) \geq  \theta \bigr]}
    \le c/\sqrt{n} ,
  \end{align*}
  for some $\mu$ that depends only on $\summ x$ and $\rho$.
  On the other hand, again by  \Cref{lem:Berry-Esseen}, for every $\theta$, we have
  \begin{align*}
    \abs[\big]{ \Pr\bigl[ B \ge \theta \sqrt{n} \bigr] - \Pr \bigl[ \mathcal{N}(0,1) \ge  \theta \bigr]}
    \le 1/\sqrt{n} . 
  \end{align*}
  Combining the above with \Cref{lem:mixture-distance-improved}, we conclude that there exists some $\theta$ such that
  \[
    \abs[\big]{\Pr\bigl[ B \ge \theta \sqrt{n} \bigr] -    \Pr\bigl[ \summ (D \cdot N_{\rho}) \ge \theta \sqrt{n}  \bigr] }
    \geq 2^{-ck/\rho} - c/\sqrt{n} . \qedhere
  \]
\end{proof}

\section{Proof of \Cref{lem:mixture-distance-improved}}

We wish to show that a linear sum of $k$ exponential functions with variance $< 1$ cannot approximate the standard normal well. We can factor the two expressions so the mixture becomes a linear sum of $k$ exponential functions, which can be written  as $\sum_{i \in [k]} a_i e^{b_i}$, while the standard normal can be written as $e^{ \alpha x^2}$. This factoring crucially uses the fact the variances in the mixture are identical.

We then argue the distance must be large for some point. To prove this, we show the entries in the inverse of a Vandermonde like matrix which corresponds to the $e^{\alpha x^2}$ are not too large. This step is the bulk of the proof. On the other hand, the Vandermonde matrix corresponding to the sum of exponentials is singular. After some matrix norm manipulations this allows us to achieve the desired result.

\begin{lemma}\label{thm:mixture_vs_normal}
Suppose that $f(x)$ is the PDF of a Gaussian distribution with variance $1$,
and $g(x)$ is the PDF of a mixture of $k$ Gaussian distributions with variance $1-\rho^2 = \sigma^2 <1$. Then 
$$
\|f-g\|_\infty\geq e^{-ck/\rho}.
$$
\end{lemma}

We let $\phi(x): = \frac{1}{\sqrt{2\pi}} e^{-x^2/2}$ denote the probability density function of $\mathcal{N}(0,1)$.

\begin{proof}[Proof of \Cref{lem:mixture-distance-improved} assuming $\Cref{thm:mixture_vs_normal}$]
Let $g(x)$ denote the PDF of $M$  and set $h(x)=\phi(x)-g(x)$. By \Cref{thm:mixture_vs_normal}, there exists some $a\in \R$ with
$$
|h(a)|\geq e^{-c k/\rho}.  
$$
We have $|\phi'(x)|\leq (2\pi e)^{-1/2}$ and $|g'(x)|\leq (2\pi e)^{-1/2}/\sigma^{2}$ for all $x$,
so $|h'(x)|\leq (2\pi e)^{-1/2}(1+\frac{1}{\sigma^2})< \frac{1}{2\sigma^2}$.
We claim there exists an interval $I\subseteq \R$ with 
\[
|\Pr [\mathcal{N}(0,1) \in I]-P[M\in I]| \geq 2\sigma^2 \left(e^{-ck/\rho}\right)^2     = e^{-ck/\rho}.
\]
To see this, assume that $h(0)\geq e^{-ck/\rho}$. Then $h(x) \geq e^{-ck/\rho} - \frac{1}{2\sigma^2} |x| $ for any $x$. So we set the interval $I = [- 2\sigma^2 e^{-ck/\rho},2\sigma^2 e^{-ck/\rho}] $, and then $\int_I h(x)dx \geq 2\sigma^2 (e^{-ck/\rho})^2$. Finally, note the assumption that $a=0$ can be made without loss of generality.  \qedhere
\end{proof}

\subsection{Proof of \Cref{thm:mixture_vs_normal}} 

The main technical result we need is the following.

\begin{lemma}\label{lem:GapMiddle}
Let $\alpha,D>0$ be fixed. Let 
$$\Delta(k):=\inf_g\|e^{\alpha x^2}-g(x)\|_\infty
$$
where the infimum is over $g$ that are a linear combination of $k$ exponential functions,
and the norm $\|\cdot\|_\infty$ is the supremum over the interval $[-D\sqrt{k},D\sqrt{k}]$. Then we have
$$
\Delta(k)\geq \exp\Big(\frac{-ck}{D^2\alpha}\Big).
$$
\end{lemma}

\begin{proof}[Proof of \Cref{thm:mixture_vs_normal} assuming \Cref{lem:GapMiddle}]
Without loss of generality we may assume that $f(x)$ is the PDF of the standard normal distribution with mean 0 and variance $1$. Define $\overline{f}(x)=e^{x^2/2\sigma^2}f(x)=e^{\alpha x^2}$ with $\alpha=\frac{1}{2\sigma^2}-\frac{1}{2}$, 
and $\overline{g}(x)=e^{x^2/2\sigma^2}g(x)$. Now $g(x)$ is a linear combination of $k$ exponential functions. If we choose some $D>0$ then \Cref{lem:GapMiddle}
gives us
\[
\Delta(k)\geq \exp\Big(\frac{-ck}{D^2 \alpha}\Big),
\]
where $\Delta(k)$ is the supremum of $|\overline{f}-\overline{g}|$
over the interval $[-D\sqrt{k},D\sqrt{k}]$.
It follows that
\[
\|f-g\|_\infty=\|e^{-x^2/2\sigma^2}(\overline{f}-\overline{g})\|_\infty\geq
\exp\Big(\frac{-D^2k}{2\sigma^2}\Big)\Delta(k)=\exp\Big(-k\Big(\frac{D^2}{2\sigma^2}+\frac{c}{D^2\alpha}\Big)\Big),
\]
Then if we set $D^2 = c\sqrt{\sigma^2/\alpha}$ we have
\[
\|f-g\|_\infty\geq \exp\Big(\frac{-ck}{\sqrt{\sigma^2\alpha}}\Big)=
\exp\Big(\frac{-ck}{\sqrt{1-\sigma^2}}\Big). \qedhere
\]
\end{proof}

\subsection{Proof of \Cref{lem:GapMiddle}}
 
\begin{definition}
For a function $f:\Z\to \R$ define the $(k+1)\times (k+1)$ matrix $M_k(f)$ by $M_k(f)_{i,j}=f(i+j-k-2)$.
For example, 
$$M_3(f)=\begin{pmatrix}
f(-3) & f(-2) & f(-1) & f(0)\\
f(-2) & f(-1) & f(0) & f(1)\\
f(-1) & f(0) & f(1) & f(2)\\
f(0) & f(1) & f(2) & f(3)
\end{pmatrix}
$$
\end{definition}

\begin{fact}\label{lem:MkSingular}
If $f(x)=\sum_{i=1}^k a_i e^{b_i x}$, then $\det M_k(f)=0$.
\end{fact}
\begin{proof}
%We can write $f(x)=\sum_{i=1}^k a_i \lambda_i^x$ where $\lambda_i=e^{\beta_i}$.
%Consider the polynomial  
%$$p(x)=(x-\lambda_1)(x-\lambda_2)\cdots (x-\lambda_k)=x^k+c_{k-1}x^{k-1}+\cdots+c_1x+c_0.$$ 
%Because $\lambda_1,\lambda_2,\dots,\lambda_k$ are roots of $p(x)$, we have
%\begin{multline*}
%0=\sum_{i=1}^k a_i\lambda_i^xp(\lambda_i)=\sum_{i=1}^ka_i
%(\lambda_i^{x+k}+c_{k-1}\lambda_i^{x+k-1}+\cdots+c_0\lambda_i^n)=\\=
%f(x+k)+c_{k-1}f(x+k-1)+\cdots+c_1f(x+1)+c_0f(x).
%\end{multline*}
%For the vector $c=[c_0\ c_1\ \cdots\ c_{k-1}\ 1]$ we get $cM_k(f)=0$, so $\det M_k(f)=0$.
We claim the sequence $\dots,f(-1),f(0),f(1),f(2),\dots$ satisfies a linear recurrence of order $k$, which implies the columns of $M_k(f)$ are linearly dependent. 

To prove the claim, we show the existence of $c_1, \dots, c_k$ such that
\[
f(x) = c_1 f(x-1)+  \dots + c_k f(x-k).
\]
Solving for $c_1,\dots, c_k$, we obtain $k$ linear constraints
\[
1 = \frac{c_1}{e^{b_1}} + \dots + \frac{c_k}{e^{k b_1}}, \ldots, 1 =  \frac{c_1}{e^{b_k}} + \dots + \frac{c_k}{e^{k b_k}}.
\]
There exists a solution to this system.
\end{proof}

\begin{fact}\label{lem:series}
If $|x|<1$ then 
$$
\prod_{i=1}^\infty(1-x^i)\geq \exp\Big(\frac{-c}{1-x}\Big).
$$
\end{fact}
\begin{proof}
Using absolute convergence, and the inequality $(1-x^j)\geq jx^j(1-x)$, we get
$$
\sum_{i=1}^\infty\log(1-x^i)=-\sum_{i=1}^\infty\sum_{j=1}^\infty \frac{x^{ij}}{j}=
-\sum_{j=1}^\infty\sum_{i=1}^\infty  \frac{x^{ij}}{j}=-\sum_{j=1}^\infty\frac{x^j}{j(1-x^j)}\geq- \sum_{j=1}^\infty\frac{1}{j^2(1-x)}=-\frac{\pi^2}{6(1-x)}.
$$
Then apply the exponential function to both sides.
\end{proof}

\begin{lemma}\label{lem:VandermondeBound}
Suppose that $q>1$ and $k$ is a positive integer. Let $A:=M_k(q^{x^2})^{-1}$. Then
$$
|A_{i,j}|\leq \frac{\displaystyle{k\choose i-1}{k\choose j-1}}{\prod_{i=1}^k(1-q^{-2i})}.
$$
\end{lemma}

\begin{proof}[Proof of \Cref{lem:GapMiddle}]
Let $f(x)=e^{\alpha x^2}$
and $g(x)=\sum_{i=1}a_ie^{b_ix}$. Define $\widetilde{f}(x)=f(Dx/\sqrt{k})=q^{x^2}$,
where $q=e^{D^2\alpha/k}$
and $\widetilde{g}(x)=g(Dx/\sqrt{k})$. By \Cref{lem:MkSingular}, $M_k(\widetilde{g})$ is singular.
Let $A$ be the inverse of $M_k(\widetilde{f})$.
The matrix
$$
AM_k(\widetilde{g})=AM_{k}(\widetilde{f})-AM_k(\widetilde{f}-\widetilde{g})=
I-A M_k(\widetilde{f}-\widetilde{g})
$$
is singular. It follows that
$$
\|A\|_\sigma\|M_k(\widetilde{f}-\widetilde{g})\|_\sigma \geq \|A M_k(\widetilde{f}-\widetilde{g})\|_\sigma\geq 1
$$
where $\|A\|_\sigma$ is the spectral norm of $A$.
The matrix $A$ is positive definite symmetric and the sum of the singular values is the sum of the eigenvalues which is equal to the trace of $A$. By \Cref{lem:VandermondeBound} we get
\[
\|A\|_\sigma\leq \operatorname{trace}(A)\leq \frac{\displaystyle\sum_{i=0}^k {k\choose i}^2}{\prod_{i=1}^k(1-q^{-2i})} = \frac{4^k}{\prod_{i=1}^k(1-q^{-2i})}.
\]
On the other hand,
\[
\|M_k(\widetilde{f}-\widetilde{g})\|_\infty
\leq (k+1)\|f-g\|_\infty.
\]
Combining everything and using \Cref{lem:series}, we get
\[
  \|f-g\|_{\infty}\geq \frac{\prod_{i=1}^k(1-q^{-2i})}{4^k(k+1)}\geq
  \exp\Big(\frac{-c}{1-q^{-2}}-\log(4)k-\log(k+1)\Big)=
  \exp\Big(\frac{-ck}{D^2\alpha}\Big). \qedhere
\]
\end{proof}

\subsection{Proof of \Cref{lem:VandermondeBound}}
First we define Vandermonde matrices.
\begin{definition}
$$
    \operatorname{Vand}(x_0,x_1,x_2, \dots,x_k)=\begin{pmatrix}
    1 & 1 & 1& \cdots & 1\\
    x_0 & x_1 & x_2 &\cdots & x_k\\
    x_0^2 & x_1^2 & x_2^2 &\cdots & x_k^2\\
    \vdots & \vdots & \vdots & & \vdots\\
    x_0^k & x_1^{k} & x_2^{k} & \cdots & x_k^k
\end{pmatrix}.$$
\end{definition}

\begin{proof}
We can transform $M_k(q^{x^2})$  into $\Vand(1,q^{2},q^{4},\dots,q^{2k})$ by multiplying the rows and columns of $M_k(q^{x^2})$ with powers of $q$. Thus by \Cref{prop:vandermonde_powers}, stated at the end,
$(-1)^{i+j}A_{i,j}\prod_{b=1}^k (q^{2b}-1)$  is a sum of ${k\choose i-1}{k\choose j-1}$ powers of $q$. This implies
$(-1)^{i+j}A_{i,j}\prod_{b=1}^k (1-q^{-2b})$ 
is also a sum of ${k\choose i-1}{k\choose j-1}$ powers of $q$. 

Next we claim no positive powers of $q$ appear in the aforementioned sum.
We define
$$
B_k(q):=\begin{pmatrix}
q^{(-k)^2/2} & & &\\
& q^{(2-k)^2/2} & & \\
 & & \ddots & \\
 & & & q^{k^2/2}
 \end{pmatrix}
$$
so that we can write
$$
M_k(q^{x^2})=B_k(q)C_k(q)B_k(q)
$$
where $C_k(q)$ is a matrix with $1$ on the diagonal and negative powers of $q$ outside of the diagonal. In particular, $C_k(q)$ converges to the identity matrix as $q\to \infty$. So 
$$A=A(q)=M_{k}(q^{x^2})^{-1}=B_k(q)^{-1}C_k(q)^{-1}B_k(q)^{-1}
$$
converges as $q\to \infty$ because both $B_k(q)^{-1}$ and $C_k(q)^{-1}$
converge. This shows that $A_{i,j} \prod_{b=1}^k (1-q^{-2b})$
cannot have positive powers of $q$ in its corresponding sum. 
Since
$$
|A_{i,j}| \prod_{b=1}^k (1-q^{-2b})
$$
is a sum of ${k\choose i-1}{k\choose j-1}$ non-positive powers of $q$ and $q>1$
we get
\[
|A_{i,j}|\prod_{b=1}^k (1-q^{-2b})\leq {k\choose i-1}{k\choose j-1}. \qedhere
\]
\end{proof}

\begin{proposition}\label{prop:vandermonde_powers}
Let $V=\operatorname{Vand}(1,q,q^2,\dots,q^k)$. Then 
\[
(-1)^{i+j} (V^{-1})_{i,j} \prod_{b=1}^k(q^b-1)      
\]
is a sum of ${k\choose i-1}{k\choose j-1}$ powers of $q$.
\end{proposition}
\begin{proof}
Note that $\det(V)=\prod_{0\leq a<b\leq k}(q^b-q^a)$.
Let $\widetilde{V}_i$ be the matrix $V$ with the $i$-th column removed,
and $\widetilde{V}_{j,i}$ be the matrix $V$ with the $j$-th row and $i$-th column removed. By the formula of $V^{-1}$ from Cramer's rule we get
$$
(V^{-1})_{i,j}=\frac{(-1)^{i+j}\det(\widetilde{V}_{j,i})}{\det(V)}.
$$
Note that $\widetilde{V}_{k+1,i}=\operatorname{Vand}(1,q,\cdots,q^{i-2},q^i,\dots,q^{k+1})$,
so 
$$\det(\widetilde{V}_{k+1,i})=\prod_{\scriptstyle 0\leq a<b\leq k\atop \scriptstyle a,b\neq i-1}(q^b-q^a).
$$
So we have
$$
\frac{ \prod_{b=1}^k(q^b-1)  \det(\widetilde{V}_{k+1,i})}{\det(V)}=\frac{ \prod_{b=1}^k(q^b-1)  }{\prod_{j=i}^{k} (q^j-q^{i-1}) \prod_{j=0}^{i-2}(q^{i-1}-q^j)}
$$
which is up to a power of $q$ factor equal to 
$$
\frac{ \prod_{j=1}^k(q^j-1)  }{\prod_{j=1}^{k-i+1} (q^j-1) \prod_{j=1}^{i-1}(q^j-1)}={k\choose i-1}_q,
$$
where the right-hand side is a Gaussian $q$-binomial coefficient which is a sum of ${k\choose i-1}$ powers of $q$. To see this, consider the generating function $\prod_{j=0}^{k-1}(1+q^jt)=\sum_{j=0}^k q^{j(j-1)/2}{k\choose j}_qt^j$. This implies that ${k\choose i-1}_q$ is a sum of ${k\choose i-1}$ 
powers of $q$.

Next, by \cref{claim:vandermonde-quotient} we have that
$$
\frac{\det(\widetilde{V}_{j,i})}{\det(\widetilde{V}_{k+1,i})}=
e_{k+1-j}(1,q,\dots,q^{i-2},q^i,\dots,q^k)
$$
is a sum of ${k\choose j-1}$ powers of $q$. We conclude that
\[
(-1)^{i+j} (V^{-1})_{i,j} \prod_{b=1}^k(q^b-1) =  \frac{\det(\widetilde{V}_{j,i})}{\det(\widetilde{V}_{k+1,i})}\cdot \frac{\prod_{b=1}^k(q^b-1)\det(\widetilde{V}_{k+1,i})}{\det(V)}.
\]
is a sum of ${k\choose i-1}{k\choose j-1}$ powers of $q$.
\end{proof}

\begin{claim} \label{claim:vandermonde-quotient}
Let $X = \operatorname{Vand}(x_0, \dots, x_k)$. Then 
\[
\frac{\det(\widetilde{X}_{j,i})}{\det(\widetilde{X}_{k+1,i})}=
e_{k+1-j}(x_0, \dots, x_{i-2}, x_{i}, \dots, x_k).
\]
\end{claim}

\begin{proof}
Without loss of generality assume $i=1$. Let $X' := \widetilde{X}_{k+1,1}$, and note $X' = \operatorname{Vand}(x_1, \dots, x_k)$.  We first sketch out the polynomial argument which proves that $\det(X' ) = \prod_{1\leq a < b \leq k} (x_b - x_a)$ (see \cite{wiki-vandermonde}).

For $b\neq a, (x_b - x_a)$ is a factor of $\det(X')$, since if we replace $x_b$ with $x_a$ in $X'$ then the determinant becomes 0. Thus, $\det(X') = p_1 \prod_{1 \leq a < b \leq k} (x_b - x_a)$ for some polynomial $p_1$. 

Now by the Leibniz formula for the determinant, since all entries of the $j$th row have degree $j-1$,  $\det(X')$ is a homogeneous polynomial of degree $1 +  \dots + k-1 =  k(k-1)/2$.  This implies that $p_1$ is a constant. 

Finally, $p_1 = 1$ since the product of the diagonal entries of $X'$ is $x_2 x_3^2 \dots x_k^{k-1}$, which is the monomial obtained by taking the first entry of each term in $\prod_{1 \leq a < b \leq k} (x_b - x_a)$.

Next we have that $\det(\widetilde{X}_{j,i})  = p_2 \prod_{1 \leq a < b \leq k} (x_b - x_a)$ for some homogeneous polynomial $p_2$ of degree $ k - (j-1)$. This follows by repeating the start of the previous argument.

Moreover we claim that $p_2$ is symmetric and squarefree. The first claim follows since if we swap $x_b$ with $x_a$ in $\widetilde{X}_{j,i}$ that will only change the sign of the determinant, and this occurs in $\prod_{1 \leq a < b \leq k} (x_b - x_a)$. The second claim follows since there are terms in  $\prod_{1 \leq a < b \leq k} (x_b - x_a)$ where the degree of an individual variable is $k-1$. Thus if $p_2$ is not square free, the degree of these variables becomes $> k$,  which is a contradiction. 

This implies that $p_2 = e_{k+1-j}(x_1, x_2, \dots, x_k)$.
\end{proof}

\paragraph{Acknowledgements.}
We thank Rocco Servedio for pointing us to \cite{Petrov-book}.

\bibliographystyle{alpha}
\bibliography{ref,OmniBib}

\end{document}